\documentclass[preprint,3p,times,onecolumn]{elsarticle}

\usepackage{paralist, enumerate, enumitem}
\usepackage{amssymb, amsmath, amsthm, bm}
\usepackage{geometry}
\usepackage{graphicx, subfigure}
\usepackage{booktabs, multirow}
\usepackage{algorithm, algorithmic}
\usepackage{mathtools}
\usepackage[retainorgcmds]{IEEEtrantools}

\graphicspath{{figures/}}
\newlength\imagewidth
\newlength\imagewidtha
\newtheorem{Proposition}{Proposition}
\newtheorem{Corollary}{Corollary}

\DeclarePairedDelimiter\abs{\lvert}{\rvert}%
\makeatletter
\let\oldabs\abs
\def\abs{\@ifstar{\oldabs}{\oldabs*}}

\begin{document}

\begin{frontmatter}

\title{Chosen-plaintext attack of an image encryption scheme based \\on modified permutation-diffusion structure}
\author[cn-xtu-cie,cityu]{Yuansheng Liu}

\author[cityu]{Leo Yu Zhang\corref{corr}}
\ead{leocityu@gmail.com}
\cortext[corr]{Corresponding author.}

\author[cityu]{Jia Wang}
\author[cn-xnu]{Yushu Zhang}
\author[cityu]{Kwok-wo Wong}

\address[cn-xtu-cie]{College of Information Engineering, Xiangtan University, Xiangtan 411105, Hunan, China}
\address[cityu]{Department of Electronic Engineering, City University of Hong Kong, HongKong SAR}
\address[cn-xnu]{School of Electronics and Information Engineering, Southwest University, Chongqing 400000, China}

\begin{abstract}
Since the first appearance in Fridrich's design, the usage of permutation-diffusion structure for designing digital image cryptosystem has been receiving increasing research attention
in the field of chaos-based cryptography. Recently, a novel chaotic Image Cipher using one round Modified Permutation-Diffusion pattern (ICMPD) was proposed.
Unlike traditional permutation-diffusion structure, the permutation is operated on bit level instead of pixel level and the diffusion is
operated on masked pixels, which are obtained by carrying out the classical affine cipher, instead of plain pixels in ICMPD.
Following a \textit{divide-and-conquer strategy}, this paper reports that ICMPD can be compromised by a chosen-plaintext attack efficiently and
the involved data complexity is linear to the size of the plain-image. Moreover, the relationship
between the cryptographic kernel at the diffusion stage of ICMPD and modulo addition then XORing  is explored thoroughly.

\end{abstract}

\begin{keyword}
Image encryption\sep Cryptanalysis \sep Chosen-plaintext attack \sep Permutation \sep Diffusion
\end{keyword}
\end{frontmatter}

\section{Introduction}
\label{sec:intro}
In the field of chaos-based cryptography, Fridrich's design \cite{fridrich1998symmetric}, we refer to it as permutation-diffusion structure in this paper, receives remarkable research attention \cite{chen2004symmetric,wong2008fast,zhu2011chaos,zhang2014chaotic,norouzi2014simple,yang2015novel}.
Inheriting from the substitution permutation network, this scheme suggests iterating the permutation and diffusion stage several rounds to earn good confusion and diffusion effect \cite{shannon1949communication},
as depicted in Fig.~\ref{fig:struture}.

Extending the work of Fridrich's can be carried out in various aspects. Chen \textit{et al.} proposed using $3$D chaotic cat map to de-correlate the relationship among pixels in the permutation stage instead of $2$D map \cite{chen2004symmetric}. Observing that one diffusion round, which typically proceeds in a sequential manner and involves nonlinear operations, often possesses higher computational complexity than that of a permutation round, Wong \textit{et al.} proposed to use a ``add-and-then-shift" strategy to include some diffusion effect in the permutation stage \cite{wong2008fast}.
In this way, the iteration round as well as the computational complexity can be reduced without affecting the security level of the resultant cryptosystem.
This idea was further studied by Zhu \textit{et al.} in \cite{zhu2011chaos} to design bit level permutation techniques.

For the sake of efficiency, there are some researchers devoted their attention to design secure chaos-based cryptosystem in the extreme case, i.e., the iteration round is only one. In \cite{zhang2014chaotic},
Zhang \textit{et al.} proposed a chaos-based image cipher based on one round permutation-diffusion structure, where some plaintext information is fed back to the key schedule.
In \cite{norouzi2014simple}, Norouzi \textit{et al.} suggested correlating the key schedule with the sum of plaintext data to construct chaotic cipher with a single diffusion round. The intuitive extension of their work is to include a permutation stage in the whole system, as suggested by Yang \textit{et al.} in \cite{yang2015novel}. In \cite{zhu2014image}, Zhu \textit{et al.} suggested a chaotic Image Cipher using one round Modified Permutation-Diffusion (ICMPD) architecture. Different from Fridrich's design,
the permutation stage is operated on bit level instead of pixel level and the diffusion stage is
operated on the output of classical affine cipher instead of plain pixel.

As stated in \cite{arroyo2013cryptanalysis}, most image cryptosystems based on one round permutation-diffusion
architecture are not secure under chosen plaintext attack (CPA) scenario. This paper reports that ICMPD suffers from the
same defect. Unlike many cryptanalysis work which only deal with specific chaos-based image
cryptosystem \cite{li2012breaking,wang2014cryptanalysis,zhang15MTAP}, this work makes several contributions.
First, we provide a quantitative security evaluation framework to both the diffusion kernel of ICMPD and the classical
modulo then XORing operation.
Second, we report that employment of the nonlinear modulo operation will inevitably
leads to the problem of the existence of (partial equivalent) key streams in the one round permutation-diffusion structure. 
Finally, the application of our result lead to an efficient
CPA attack to ICMPD. This is a reproducible research and all the codes are openly
accessible\footnote{https://sites.google.com/site/leoyuzhang/.}.

The rest of the paper is organized as follows. The next section describes the details of ICMPD and then provides
some experimental results for illustration.
In Sec.~3, the diffusion kernel of ICMPD is casted to the form of modulo then XORing and analyzed thoroughly.
Sec.~4 explains how to break ICMPD using a \textit{divide-and-conquer strategy} in CPA scenario, followed by some
simulation results. The last section concludes our work by briefly discussing the possible remedies of ICMPD.

\section{The image encryption scheme under study}
The image encryption scheme proposed in \cite{zhu2014image}, i.e., ICMPD, is applied to gray-scale image with $L = H \times W$
pixels. It exploits the permutation-diffusion structure suggested by Fridrich \cite{fridrich1998symmetric} with the following two
modifications: a) the permutation is operated on bits instead of pixels; b) the diffusion is operated on masked pixels instead of
plain pixels. For the sake of clarity, we depict the schematic diagram of ICMPD in Fig.~\ref{fig:newstruture}
and modify the notations used in \cite{zhu2014image} to describe the scheme under study.

\subsection{Key schedule}
\label{sec:key}
The secret key of ICMPD is composed of a set of initial values and control parameters for several chaotic systems. Specifically, they are:
\begin{itemize}
\item Initial value $(x_0, y_0)$ and control parameters $(a, b)$ of the following generalized Arnold map
    \begin{equation}
    \left(
    \begin{array}{c}
    x_{n+1}   \\
    y_{n+1}
    \end{array} \right)
    =
    \left(
    \begin{array}{cc}
    1  & a \\
    b  & 1+ab
    \end{array} \right)
    \left(
    \begin{array}{c}
    x_{n}   \\
    y_{n}
    \end{array} \right)  \bmod 1,
    \label{eq:genArn}
    \end{equation}
where $a>1$, $b>1$ and $(x \bmod 1)$ represents the fractional part of real number $x$.

\item Two sets of initial value and control parameter, i.e., $\{(k', x'_0), (k^\diamond, x^\diamond_0)\}$, of the following Chebyshev map
    \begin{equation}
    x_{n+1} = \cos (k \cdot \arccos (x_n)),
    \label{eq:chebyshev}
    \end{equation}
where $k \leq 2$ and $x_n \in [-1, 1]$.

\item Initial value and control parameter $(\mu, x^*_0)$ of the following Logistic map
    \begin{equation}
    x_{n+1} = \mu x_n (1-x_n) ,
    \label{eq:logistic}
    \end{equation}
where $\mu \in (3.57, 4)$ and $x_n \in (0, 1)$.
\end{itemize}
The secret key streams employed in the row/column permutation stage, substitution stage and diffusion stage are obtained through post-processing the chaotic systems orbits.
These processes can be summarized as follows:
\begin{enumerate}
\item \textit{Permutation streams $E_r$ and $E_c$.}
Iterate the generalized Arnold map (\ref{eq:genArn}) using the partial key $(x_0, y_0, a, b)$ $h_0 + 8L$ times and denote the latter $8L$ outputs by $X = \{x_i\}_{i=1}^{8L}$ and $Y = \{y_i\}_{i=1}^{8L}$. Sort $X$ and $Y$ in ascending order
and get the permutation streams $E_r = \{e_r({i})\}_{i=1}^{8L}$ and $E_c = \{e_c({i})\}_{i=1}^{8L}$ by comparing $X$ and $Y$ with their sorted versions, respectively.

\item \textit{Substitution streams $S$ and $T$.}
Run the Chebyshev map (\ref{eq:chebyshev}) iteratively through $(k', x'_0)$ and post-process the resultant orbit $x_i$ by
\begin{equation}
y_i = \lfloor  10^9 \cdot \abs{x_i} \rfloor \bmod 256,
\label{eq:postpro}
\end{equation}
where $\abs{x}$ and $\lfloor x \rfloor$ return the absolute value of $x$ and the largest value not larger than $x$, respectively. If $\gcd(y_i, 256) =1$, we push this value to $S$. Otherwise, we proceed with the next orbit $x_{i+1}$ till the length of $S$ reaches $L$. Finally, it comes to the conclusion that we obtain a random number stream $S = \{s(i)\}_{i=1}^L$, whose elements are coprime to $256$.
Similarly, run Eq.~(\ref{eq:chebyshev}) under $(k^\diamond, x^\diamond_0)$ and get $\{x^\diamond_i\}_{i=1}^{L}$. Quantize the result using Eq.~(\ref{eq:postpro}) and obtain
$T = \{t(i)\}_{i=1}^{L}$.

\item \textit{Diffusion stream $R$.}
Execute the Logistic map (\ref{eq:logistic}) under $(\mu, x^*_0)$ iteratively and obtain random chaotic orbits $\{x^*_i\}_{i=1}^{L}$. Quantize the sequences by Eq.~(\ref{eq:postpro}) and denote the results by $R=\{r(i)\}_{i=1}^{L}$.

\end{enumerate}

\subsection{Encryption process}
\label{sec:encpro}
As depicted in Fig.~\ref{fig:newstruture}, the encryption process in ICMPD
is composed of the following steps:
\begin{enumerate}
\item \textit{Bit decomposition.} Scan an image $P$ in the raster order and  obtain a pixel sequence
$\{p(i)\}_{i=1}^{L}$. Decompose each pixel of $P$ to its $8$ bits and denote the binary sequence by
$B=\{b(j)\}_{j=1}^{8L}$, where
$p(i)=\sum_{k=1}^8 b(8(i-1)+k)\cdot 2^{k-1}$.

\item \textit{Bit permutation\footnote{For simplicity, we slightly modify the permutation techniques described in \cite{zhu2014image} while keeping its security level unchanged.}.}
    Permute the binary format of the image $B$ in both horizontal and vertical directions and get
    $\bar{B}= \{\bar{b}(j)\}_{j=1}^{8L}$ via
    \begin{equation}
    \bar{b}(j) = b(e_c(e_r(j))).
    \label{eq:bitpermuation}
    \end{equation}
\item \textit{Local pixel substitution.} Combine every $8$-bit of $\bar{B}$ to a new pixel sequentially using
    \begin{equation}
    p'(i) = \sum_{k=1}^8 \bar{b}(8(i-1)+k) \cdot 2^{k-1}
    \label{eq:bitcomposition}
    \end{equation}
    where $i = 1 \sim L$. The obtained pixels are substituted using the affine cipher orderly, i.e.,
    \begin{equation}
    c'(i) = p'(i)s(i) \dotplus t(i),
    \label{eq:affine}
    \end{equation}
    where $a\dotplus b = (a+b) \bmod 2^8$.

\item \textit{Global pixel diffusion.} Collect the substitution result $C' = \{c'(i)\}_{i=1}^{L}$
and update it by the classical diffusion rule as follows
    \begin{equation}
    c(i) = c'(i) \oplus r(i) \oplus c(i-1),
    \label{eq:diffusion}
    \end{equation}
    where $c(0)=172$ and $i = 1\sim L$. Finally, transform the ciphertext sequence
    $C= \{c(i)\}_{i=1}^{L}$ into an image of size $H\times W$.
\end{enumerate}

The decryption can be achieved by executing the encryption steps reversely, detailed description can be found in \cite[Sec.~3]{zhu2014image}.
As demonstrated by Zhu \textit{et al.} in \cite[Sec.~4]{zhu2014image}, the new scheme should possess high security since:
1) the key space is large enough to resist brute-force attack; 2) the adoption of multiple chaotic systems
for the generation of key streams guarantees good key sensitivity; 3) the modified permutation-diffusion architecture introduces diffusion effect in both permutation and diffusion stage, which may frustrate any plaintext attacks.
For illustration purpose, we set the secret key $(x_0, y_0, a, b, k', x'_0, k^\diamond, x^\diamond_0, \mu, x^*_0)$ to
\begin{small}
$(0.346, 0.478, 1.644, 2.986, 4.434, 0.6435, 5.673, 0.523, 3.14, 0.34)$.
\end{small}
Two $128\times 128$ plain-images, ``Lena" and ``Peppers", shown in Fig.~\ref{fig1:plain1} and Fig.~\ref{fig1:plain2} are encrypted and their corresponding cipher-images
are depicted in Fig.~\ref{fig1:cipher1} and Fig.~\ref{fig1:cipher2}.

As we will discuss in the next section, the local pixel substitution and global pixel diffusion, which serves as the core of the nonlinear diffusion stage of the modified architecture, can be treated as the generalization of a typical modulo addition then XORing operation and is fragile in chosen-plaintext attack (CPA) scenario.
Based on this finding, a CPA is readily to compromise the cipher under study using the \textit{divide-and-conquer} strategy.

\section{Related work and main results}
\label{sec:mainresult}
The modulo addition then XORing operation, which is nonlinear and has low computational complexity,
serves as the fundamental or even the only component in many image cryptosystems \cite{chen2004symmetric,CH:HCKBA:IJBC10,Rao:ModifiedCKBA:ICDSP07,chen2014fast,zhang2015color,zhu2011chaos}.
Mathematically, it can be expressed as
\begin{equation}
c(i) = (p(i)\dotplus k(i)) \oplus k(i) \oplus c(i-1) ,
\label{eq:diffusionA}
\end{equation}
where $k(i)$ is the $i$-th element of the key stream $K$,
$p(i)$ and $c(i)$ are the $i$-th pixel of plain-image $P$ and cipher-image $C$, respectively.
Under the CPA assumption, where an adversary is able to obtain ciphertexts of arbitrary plaintexts adaptively,
the relationship of the difference between two groups of chosen plain-image and cipher-image pairs, i.e., $(P, C)$
and $(\tilde{P}, \tilde{C})$, ca be derived as follows:
\begin{equation*}
(c(i)\oplus c(i-1)) \oplus (\tilde{c}(i)\oplus \tilde{c}(i-1))
= (p(i)\dotplus k(i)) \oplus (\tilde{p}(i) \dotplus k(i)),
\end{equation*}
where $i = 1\sim L$.
More generally, we write it as
\begin{equation}
y = (\alpha \dotplus k) \oplus (\beta \dotplus k).
\label{eq:ModAdd}
\end{equation}
From the cryptanalysis point of view, these questions arise naturally:
\begin{enumerate}
\item Given a large quantities of $(\alpha, \beta, y)$, it is obvious that the exact key $k$ used for encryption will satisfy all the resultant equations of the form (\ref{eq:ModAdd}). But
 is this $k$ unique or not? This relates to the question of the existence of equivalent key.
\item How many queries of $(\alpha, \beta)$ are sufficient to recover the exact secret key $k$ or its equivalent form\footnote{The adversary can choose $(\alpha, \beta)$ freely and  be aware of the value of $y$ in CPA assumption.}?
    This relates to the resistance of the cryptosystem in CPA scenario.
\end{enumerate}
In \cite{li2011breaking}, Li \textit{et al.} proved that $3$ pairs of queries $(\alpha, \beta)$ are sufficient to
solve Eq.~(\ref{eq:ModAdd}) in terms of modulo $2^7$. Soon, they improved this result in terms of required number of queries to $2$ in \cite{li2013breaking}.

Before we dive into the detail of the proof, we would like to cast the diffusion process of ICMPD as the form of Eq.~(\ref{eq:ModAdd}).
Combining Eq.~(\ref{eq:affine}) and (\ref{eq:diffusion}), we can get
\begin{equation}
c(i) =  [  p'(i)s(i) \dotplus t(i) ] \oplus r(i) \oplus c(i-1).
\label{eq:newdiffusion}
\end{equation}
Similarly, we calculate the difference of two groups of chosen plain-image and cipher-image as follows:
\begin{equation*}
(c(i)\oplus c(i-1)) \oplus (\tilde{c}(i)\oplus \tilde{c}(i-1))
= (p'(i)s(i)\dotplus t(i)) \oplus (\tilde{p}'(i)s(i)\dotplus t(i)) ,
\end{equation*}
Assuming $E_r$ and $E_c$ are known in advance by the adversary (or simply treat them as identity permutations), we can generalize
the above equation as
\begin{equation}
y = (\alpha s \dotplus t) \oplus (\beta s \dotplus t),
\label{eq:newmodAdd}
\end{equation}
where $s, t$ are two unknowns, $y$ is known and $\alpha, \beta$ are known and can be chosen freely by the
adversary in CPA scenario. Now, the same questions arise for Eq.~(\ref{eq:newmodAdd}). We will answer them in the following sections.


\subsection{Previous work}
The following two propositions solve the two questions related to Eq.~(\ref{eq:ModAdd}).
\begin{Proposition}
Let $\hat{k} = k\oplus 2^7$, then $\hat{k}$ is a solution of Eq.~(\ref{eq:ModAdd}) if
$k$ satisfies $y = (\alpha \dotplus k) \oplus (\beta \dotplus k)$.
\label{pro:eqkey1}
\end{Proposition}
\begin{proof}
To prove this proposition, we first examine the relationship of $k \oplus 2^7$ and $k \dotplus 2^7$.
If $k \geq 2^7$, then it is easy to conclude $k \oplus 2^7 = k - 2^7 = k \dotplus 2^7$. Similarly, we have
$k \oplus 2^7 = k + 2^7 = k \dotplus 2^7$ when $k \leq 2^7$. Therefore,
\begin{IEEEeqnarray}{rCl}
y &=& (\alpha \dotplus k ) \oplus (\beta \dotplus k) \nonumber \\
  &=& (\alpha \dotplus k ) \oplus 2^7 \oplus (\beta \dotplus k) \oplus 2^7  \nonumber \\
  &=& (\alpha \dotplus k  \dotplus 2^7) \oplus (\beta \dotplus k \dotplus 2^7)\nonumber  \\
  &=& [\alpha \dotplus (k \oplus 2^7)] \oplus [\beta \dotplus (k \oplus 2^7)]\nonumber  \\
  &=& (\alpha \dotplus \hat{k} ) \oplus (\beta \dotplus \hat{k}).\nonumber
\end{IEEEeqnarray}
Hence completes the proof.
\end{proof}
Applying this proposition directly, we can easily conclude that all the image cryptosystems employing diffusion Eq.~(\ref{eq:diffusionA}) are subjected to the problem of existence of equivalent key (stream).
To be more precisely, this problem stems from the nature of the modulo operator, i.e., the carry bit generated by the highest bit plane is discarded after the modulo operation.
In the following proposition, we answer the question of how many pairs of chosen plain-images and cipher-images, hence $(\alpha, \beta)$ can be
chosen freely and $y$ is known, are sufficient to recover the key stream of Eq.~(\ref{eq:diffusionA}) in terms of modulo $2^7$.

\begin{Proposition}
Two groups of $(\alpha, \beta)$ are sufficient to solve Eq.~(\ref{eq:ModAdd}) in terms
of modulo $2^7$. Specifically, they are $(0, 170)$ and $(170, 85)$.
\label{pro:solve1}
\end{Proposition}
\begin{proof}
The proof presented in \cite{li2011breaking,li2013breaking} involves theoretically studying the carry bit of all bit planes of Eq.~(\ref{eq:ModAdd}), details can be found in \cite[Sec.~3.3]{li2013breaking}.
Here, we would rather follow a straightforward logic to verify this proposition, which is shown to be useful for our new model Eq.~(\ref{eq:newmodAdd}).

Let $(\alpha_1, \beta_1) = (0, 170)$ and $(\alpha_2, \beta_2) = (170, 85)$, the proposition can be reformulated as
\begin{equation}
\left\{
\begin{IEEEeqnarraybox}[][c]{rCl}
\IEEEstrut
y_1 &=& (\alpha_1 \dotplus k ) \oplus (\beta_1 \dotplus k) , \\
y_2 &=& (\alpha_2 \dotplus k ) \oplus (\beta_2 \dotplus k) ,
\IEEEstrut
\end{IEEEeqnarraybox}
\right.
\label{eq:solveModAdd}
\end{equation}
where $y_1, y_2 \in [0, 255]$ are two known integers. This problem converts to whether
the solution to Eq.~(\ref{eq:solveModAdd}) is unique in terms of modolu $2^7$ given $y_1$ and  $y_2$.
More precisely, there is a unique solution for certain known $(y_1, y_2)$ tuple and
there are totally $2^7$ out of all the possible ($256 \times 256$) tuples of $(y_1, y_2)$ which leads to this unique solution.
The following procedures demonstrate how this statement is verified.
\begin{description}[noitemsep, nolistsep]
\item[Step 1:] Let $y_1 =0$, and find all the $k \in [0, 127]$ that satisfy the equation $y_1 = (\alpha_1 \dotplus k) \oplus (\beta_1 \oplus k)$ and denote them as $\mathbb{K}_{y_1}$.
\item[Step 2:] Let $y_2 =0$, and find all the $k \in \mathbb{K}_{y_1}$ that satisfy the equation
    $y_2 = (\alpha_2 \dotplus k) \oplus (\beta_2 \oplus k)$ and denote the possible results as $\mathbb{K}_{y_2}$.
\item[Step 3:] If $\#\{\mathbb{K}_{y_2}\}$ equals $1$ and $y_2 < 256$, then set $y_2 = y_2+1$ and go to Step 2. 
\item[Step 4:] Let $y_1= y_1 +1$ if $y_1<256$ and set $y_2=0$, go to Step~1.
\end{description}
Finally, we can easily obtain $128$ out of $256\times 256$ tuples of $(y_1, y_2)$ and their corresponding
$k$ from the above procedures and then construct a table composed of these $128$ triples $(y_1, y_2, k)$.
The solution of Eq.~(\ref{eq:ModAdd}) under queries $(\alpha_1, \beta_1)$ and
$(\alpha_2, \beta_2)$ can be determined by simple look-up-table, hence finishes the proof of the
proposition.
\end{proof}
Proposition~\ref{pro:solve1} deals with the problem of finding the solution of Eq.~(\ref{eq:ModAdd}), and thus determining the diffusion key stream
of Eq.~(\ref{eq:diffusionA}) in the context of a CPA scenario.
Instead of studying all the carry bits of Eq.~(\ref{eq:ModAdd}) theoretically, the proof shown above heavily relies on
exhaustively search over all the $256 \times 256$ combinations. This makes the proof seem informal but it possesses the following
advantages: a) It is extremely fast since the number of the combinations is only $256 \times 256$; b) The by-product, i.e., the
table composed of $128$ triples $(y_1, y_2, k)$, allows one find the key stream for Eq.~(\ref{eq:diffusionA}) by a trivial look-up-table operation;
c) It can be easily extended to other diffusion operations when theoretically studying all the carry bits is
difficult, if not impossible.

\subsection{Main results}
Based on the strategy presented above, we answer the questions about the solution of
Eq.~(\ref{eq:newmodAdd}) in the following.

\begin{Proposition}
Suppose $y, s, \alpha, \beta \in [0, 255]$, $t\in [0, 128)$ and $\gcd(s, 256) =1$.
Given $\alpha$, $\beta$ and $y$, the equation
$y = (\alpha s \dotplus t) \oplus (\beta s \dotplus t)$ has four equivalent solutions.
Specifically, they are $(s,t)$, $(s, t+128)$, $(256 -s, 127-t)$ and $(256-s, 255-t)$.
\label{pro:eqkey2}
\end{Proposition}
\begin{proof}
Let $f(s, t) =(\alpha s \dotplus t) \oplus (\beta s \dotplus t)$, the proposition is proved if the following three equations are true:
\begin{enumerate}[label={(\roman*)},noitemsep, nolistsep]
\item $f(s, t) = f(s, t+128)$;
\item $f(s, t) = f(256-s, 127-t)$;
\item $f(s, t) = f(256-s, 255-t)$.
\end{enumerate}
Referring to Proposition~1, we have
\begin{IEEEeqnarray}{rCl}
f(s, t+128)&=& [(\alpha s \bmod 256) \dotplus (t+128 \bmod 256)]  \nonumber \\
           & & \oplus \: [(\beta s \bmod 256) \dotplus (t+128 \bmod 256)] \nonumber \\
           &=& [(\alpha s \bmod 256) \dotplus (t\oplus 128)]  \nonumber \\
           & & \oplus \: [(\beta s \bmod 256) \dotplus (t\oplus 128)] \nonumber \\
           &=& (\alpha s \dotplus t \dotplus 128) \oplus (\beta s \dotplus t \dotplus 128)  \nonumber \\
           &=& (\alpha s \dotplus t) \oplus  128 \oplus (\beta s \dotplus t) \oplus  128  \nonumber \\
           &=& f(s, t). \nonumber
\end{IEEEeqnarray}
To prove equation~(ii), we first consider the following two cases:
\begin{enumerate}[label={(\alph*)},noitemsep, nolistsep]
\item If $\alpha s \dotplus t < 128$, then we have
\begin{IEEEeqnarray}{rCl}
[127 -(\alpha s \dotplus t)] \bmod 256&=& 127 -(\alpha s \dotplus t)  \nonumber \\
                                      &=& (1111111)_2 - (\alpha s \dotplus t) \nonumber \\
                                      &=& (1111111)_2 \oplus (\alpha s \dotplus t) \nonumber \\
                                      &=& 127 \oplus (\alpha s \dotplus t), \nonumber
\end{IEEEeqnarray}
where $(\cdot)_2$ denotes the binary format of the operand.

\item If $\alpha s \dotplus t \geq 128$, then we have
\begin{IEEEeqnarray}{rCl}
[127 -(\alpha s \dotplus t)] \bmod 256&=& 127 + 256 -(\alpha s \dotplus t)  \nonumber \\
                                      &=& (101111111)_2 - (\alpha s \dotplus t) \nonumber \\
                                      &=& (1111111)_2 \oplus (\alpha s \dotplus t) \nonumber \\
                                      &=& 127 \oplus (\alpha s \dotplus t).  \nonumber
\end{IEEEeqnarray}
\end{enumerate}
Now, it is clear that
\begin{IEEEeqnarray}{rCl}
f(256-s, 127-t) &=& [(\alpha (256-s) \bmod 256) \dotplus (127-t)]  \nonumber \\
                & & \oplus \: [(\beta (256-s) \bmod 256) \dotplus (127-t)] \nonumber \\
                &=& [127 - (\alpha s \dotplus t)]  \bmod 256 \nonumber \\
                & & \oplus \: [127 - (\beta s \dotplus t)]  \bmod 256 \nonumber \\
                &=& 127 \oplus (\alpha s \dotplus t) \oplus 127 \oplus (\beta s \dotplus t) \nonumber \\
                &=& f(s, t).
\end{IEEEeqnarray}
Referring the result of equation~(i) and (ii), we conclude
\begin{IEEEeqnarray}{rCl}
f(256-s, 255-t) &=& f(256-s, 127-t+128) \nonumber \\
                &=& f(256-s, 127-t) \nonumber \\
                &=& f(s, t) \nonumber.
\end{IEEEeqnarray}
Finally, the proposition is proved.
\end{proof}

Apply this proposition directly, it is easy to conclude that the image cryptosystem under study, i.e., ICMPD, also suffers
from the problem of existence of equivalent key (stream).
Once again, we emphasize that this security defect is rooted from the use of modulo operation, where information of the highest carry bit is lost.

\begin{Proposition}
Suppose $y, s, t, \alpha, \beta \in [0, 255]$ and $\gcd(s, 256) =1$.
Seven groups of $(\alpha, \beta)$ are sufficient to solve the equation
\begin{IEEEeqnarray}{rCl}
y = (\alpha s \dotplus t) \oplus (\beta s \dotplus t)  \nonumber
\end{IEEEeqnarray}
in terms of modulo $2^7$. Specifically,
they are $(2^0, 2^1)$, $(2^1, 2^2)$, $(2^2, 2^3)$, $(2^3, 2^4)$,
$(2^4, 2^5)$, $(2^5, 2^6)$ and $(2^6, 2^7)$.
\label{pro:solve2}
\end{Proposition}
\begin{proof}
Theoretically studying all the carry bits becomes extremely difficult in this context as Eq.~(\ref{eq:newmodAdd}) involves a multiplication. Let
 $(\alpha_1, \beta_1)=(2^0, 2^1)$,
$(\alpha_2, \beta_2)=(2^1, 2^2)$, $(\alpha_3, \beta_3)=(2^2, 2^3)$, $(\alpha_4, \beta_4)=(2^3, 2^4)$,
$(\alpha_5, \beta_5)=(2^4, 2^5)$, $(\alpha_6, \beta_6)=(2^5, 2^6)$ and $(\alpha_7, \beta_7)=(2^6, 2^7)$,
the problem turns to whether the following system of equations has
a single unique solution in terms of modulo $2^7$
for certain known integers $y_1, y_2, \cdots, y_7 \in [0, 255]$:
\begin{equation}
\left\{
\begin{IEEEeqnarraybox}[][c]{rCl}
\IEEEstrut
y_1 &=& (\alpha_1 s \dotplus t ) \oplus (\beta_1 s \dotplus t) , \\
y_2 &=& (\alpha_2 s \dotplus t ) \oplus (\beta_2 s \dotplus t) , \\
y_3 &=& (\alpha_3 s \dotplus t ) \oplus (\beta_3 s \dotplus t) , \\
y_4 &=& (\alpha_4 s \dotplus t ) \oplus (\beta_4 s \dotplus t) , \\
y_5 &=& (\alpha_5 s \dotplus t ) \oplus (\beta_5 s \dotplus t) , \\
y_6 &=& (\alpha_6 s \dotplus t ) \oplus (\beta_6 s \dotplus t) , \\
y_7 &=& (\alpha_7 s \dotplus t ) \oplus (\beta_7 s \dotplus t) .
\IEEEstrut
\end{IEEEeqnarraybox}
\right.
\label{eq:solvePro4}
\end{equation}
The intuitive method to verify this statement is to exhaustively search all the combinations of
all $2^{56}$ $7$-tuples $(y_1, y_2, \cdots, y_7)$ using the similar procedures as described in proposition~2. This involved complexity is equal to
searching the key space of DES algorithm, which is known as computational expensive.

Observing that the unique solution $(s, t)$ is determined by $64 \times 128$ out of $2^{56}$ $7$-tuples
$(y_1, y_2, \cdots, y_7)$, we can alternatively search $64 \times 128$ possible combination
of $(s, t)$ and check whether the resultant $7$-tuple $(y_1, y_2, \cdots, y_7)$ is unique.
The following procedure verifies this assumption. 
\begin{description}[noitemsep, nolistsep]
\item[Step 1:] Let $s=1$, $t=0$ and set $\mathbb{Y}_{1-7} = \emptyset$.

\item[Step 2:] Calculate $(y_1, y_2, \cdots, y_7)$ according to Eq.~(\ref{eq:solvePro4}) under known
$s, t$ and $7$ groups of $(\alpha, \beta)$.
If the $7$-tuple $(y_1, y_2, \cdots, y_7) \notin \mathbb{Y}_{1-7}$, then add
$(y_1, y_2, \cdots, y_7)$ to the set $\mathbb{Y}_{1-7}$. Otherwise, the proposition is false.
\item[Step 3:] Let $t= t+1$ if $t<128$, go to Step~2.
\item[Step 4:] Let $s= s+2$ if $s<128$ and set $t=0$, go to Step~2.
\end{description}
Finally, one can obtain a table composed of $(64 \times 128)$ $9$-tuples, i.e.,
$(y_1, y_2, \cdots, y_7, s, t)$. Finding the solution of Eq.~(\ref{eq:newmodAdd})
under seven queries of $(\alpha, \beta)$ simplifies to look-up-table, just as we did in proposition~2.
\end{proof}

\begin{Corollary}
The solution of the equation
\begin{IEEEeqnarray}{rCl}
y = (\alpha s \dotplus t) \oplus (\beta s \dotplus t)  \nonumber
\end{IEEEeqnarray}
in terms of modulo $2^7$ can be determined by the following $8$ groups of queries: $(2^0, 0)$, $(2^1, 0)$, $(2^2, 0)$, $(2^3, 0)$, $(2^4, 0)$
$(2^5, 0)$, $(2^6, 0)$ and $(2^7, 0)$.
\end{Corollary}
\begin{proof}
It is easy to get the result with the observation that Eq.~(\ref{eq:solvePro4}) is included in the equations that are constructed from these $8$ queries.
Following the same procedures above, we construct a table of size
$(8192\times 1)$, each of whose entry is
an unique $10$-tuple $(y_1, y_2, \cdots, y_8, s, t)$. Once again, finding the solution becomes a look-up-table operation.
\end{proof}

\section{Chosen-plaintext attack of ICMPD}
As we can observe from Sec.~\ref{sec:key}, the key streams $E_r, E_c, S, T$ and $R$ are produced
independently from the encryption process. Moreover, the whole encryption is composed of a single round
(modified) permutation and diffusion. These facts can be employed to facilitate a \textit{divide-and-conquer} attack, where the whole system is cracked by employing that some
bottom-line chosen plain-images are neutral with respect to the permutation stage. For convenience, let
$u(j) = e_c(e_r(j))$ for $j \in [1, 8L]$ and denote $U=\{u(j)\}_{j=1}^{8L}$. We explain the detail of
how to recover the key streams $U, S, T$ and $R$ under a CPA scenario in the following.

\subsection{Revealing the permutation and equivalent substitution key streams ($U, S$ and $T$)}
Referring to step~1 of the encryption process (see Sec.~\ref{sec:encpro}), the intermediate
binary sequences can be obtained from the plain-image without any secret key, which allows us have the freedom to choose the binary sequences directly.

Let $B= \{b(j) \equiv 0\}_{j=1}^{8L}$ be a binary sequence with constant value $0$.
Referring to Eqs.~(\ref{eq:bitpermuation}), (\ref{eq:bitcomposition}) and (\ref{eq:newdiffusion}),
the resultant cipher-image $C=\{c(i)\}_{i=1}^{L}$ will satisfy
\begin{equation}
\left\{
\begin{IEEEeqnarraybox}[][c]{rCl}
\IEEEstrut
c(i) \oplus c(i-1) =&  [  p'(i)s(i) \dotplus t(i) ] \oplus r(i),  \\
p'(i) =& \sum_{k=1}^8 b( u(8(i-1)+k) ) \cdot 2^{k-1},
\IEEEstrut
\end{IEEEeqnarraybox}
\right.
\label{eq:chosenall0}
\end{equation}
where $i \in [1, L]$, $k \in [1, 8]$ and $b( u(8(i-1)+k) )= b(u(j)) \equiv 0$.
Now, it becomes clear that recovering $u(j)$, and then $U$, is equal to the problem of identifying
the relationship between $j$ and $(i,k)$ for all $j \in [1, 8L]$.

Slightly modify a single bit of the chosen plain binary sequence $B$, for example, set the lowest bit of the first pixel to $1$ and keep
the remaining $8L-1$ bits unchanged. Denote the modified version of $B$ as $\tilde{B}_1 = \{b_1(j)\}_{j=1}^{8L}$ and obtain its corresponding cipher-image $\tilde{C}_1=\{c_1(i)\}_{i=1}^L$.
Similar to Eq.~(\ref{eq:chosenall0}), we conclude
\begin{equation}
\left\{
\begin{IEEEeqnarraybox}[][c]{rCl}
\IEEEstrut
\tilde{c}_1(i) \oplus \tilde{c}_1(i-1) =&  [  p_1'(i)s(i) \dotplus t(i) ] \oplus r(i) ,  \\
p_1'(i) =& \sum_{k=1}^8 b_1( u(8(i-1)+k) ) \cdot 2^{k-1},
\IEEEstrut
\end{IEEEeqnarraybox}
\right.
\label{eq:chosena1}
\end{equation}
where  $b_1( u(8(i-1)+k) ) = b(u(j)) \equiv 0$ for $j>1$ and $b_1(1)=1$.
Combining Eqs.~(\ref{eq:chosenall0}) and (\ref{eq:chosena1}), it is concluded that
\begin{IEEEeqnarray}{rCl}
\lefteqn{ (c(i)\oplus c(i-1)) \oplus (\tilde{c}_1(i)\oplus \tilde{c}_1(i-1)) }  \nonumber \\
&=&   [  p'(i)s(i) \dotplus t(i) ] \oplus [  p_1'(i)s(i) \dotplus t(i) ]   \nonumber \\
&=& \left\{
            \begin{array}{rl}
             0 & \text{~~if~}i< i_1,   \\
             (0\cdot s(i) \dotplus t(i) ) \oplus (  2^{k_1}\cdot s(i) \dotplus t(i))  & \text{~~if~}i = i_1,~  \\
             0 & \text{~~if~}i> i_1,
            \end{array} \right.
            \label{eq:differencei0}
\end{IEEEeqnarray}
where $i_1 =\lfloor u^{-1}(1)/8 \rfloor +1$, $k_1 = u^{-1}(1) \bmod 8 $ and  $u(u^{-1}(i)) \equiv i$.

Given the secret key $(x_0, y_0, a, b, k', x'_0, k^\diamond, x^\diamond_0, \mu, x^*_0)$ =
\begin{small}
\small{$(0.346, 0.478, 1.644, 2.986, 4.434, 0.6435, 5.673, 0.523, 3.14, 0.34)$},
 \end{small}
 which is exactly the same as that used in \cite{zhu2014image}, we verify this statement by carrying out experiment to plain-image of size $128 \times 128$.
For illustration purpose, the cipher-images $C$ and $\tilde{C}_1$ are altered using
\begin{equation*}
\left\{
\begin{IEEEeqnarraybox}[][c]{rCl}
\IEEEstrut
v(i) &=& c(i) \oplus c(i-1)  \\
\tilde{v}_1(i) &=& c_1(i) \oplus c_1(i-1)
\IEEEstrut
\end{IEEEeqnarraybox}
\right.
\end{equation*}
and the results are denoted as $V$ and $\tilde{V}_1$. Figs.~\ref{fig:all0_c0} and \ref{fig:all0_c1} depict the cipher-image sequences corresponding to $V$ and $\tilde{V}_1$, respectively. The difference between $V$ and $\tilde{V}_1$ is shown in Fig.~\ref{fig:all0_xor}. Now, it is clear that the relationship between $j=1$ and $i=i_1$ can be readily identified.

Repeat this experiment for all the remaining bit locations, i.e., $j= 2\sim 8L$, of $B$, then one can obtain the mapping between $j \in [1, 8L]$ and $i \in [1, L]$ in the same way.

To reveal the exact permutation key stream $U$, the left problem is to identify the relationship between $j$ and $k$.
To study this problem, we set $i=i_1$ and review Eq.~(\ref{eq:differencei0})
\begin{IEEEeqnarray}{rCl}
y_1(i_1) &=& (c(i_1)\oplus c(i_1-1)) \oplus (\tilde{c}_1(i_1)\oplus \tilde{c}_1(i_1-1)) \nonumber \\
                &=& (0\cdot s(i_1) \dotplus t(i_1) ) \oplus (  2^{k_1}\cdot s(i_1) \dotplus t(i_1)). \nonumber
\end{IEEEeqnarray}
Noting that the relationship between $j\in [1, 8L]$ and $i \in [1, L]$ is revealed, we can obtain the following system of equations
\begin{equation*}
\left\{
\begin{IEEEeqnarraybox}[][c]{rCl}
\IEEEstrut
y_2(i_1) &=& (0\cdot s(i_1) \dotplus t(i_1) ) \oplus (  2^{k_2}\cdot s(i_1) \dotplus t(i_1)) , \\
y_3(i_1) &=& (0\cdot s(i_1) \dotplus t(i_1) ) \oplus (  2^{k_3}\cdot s(i_1) \dotplus t(i_1)) , \\
y_4(i_1) &=& (0\cdot s(i_1) \dotplus t(i_1) ) \oplus (  2^{k_4}\cdot s(i_1) \dotplus t(i_1)) , \\
y_5(i_1) &=& (0\cdot s(i_1) \dotplus t(i_1) ) \oplus (  2^{k_5}\cdot s(i_1) \dotplus t(i_1)) , \\
y_6(i_1) &=& (0\cdot s(i_1) \dotplus t(i_1) ) \oplus (  2^{k_6}\cdot s(i_1) \dotplus t(i_1)) , \\
y_7(i_1) &=& (0\cdot s(i_1) \dotplus t(i_1) ) \oplus (  2^{k_7}\cdot s(i_1) \dotplus t(i_1)) , \\
y_8(i_1) &=& (0\cdot s(i_1) \dotplus t(i_1) ) \oplus (  2^{k_8}\cdot s(i_1) \dotplus t(i_1)) ,
\IEEEstrut
\end{IEEEeqnarraybox}
\right.
\end{equation*}
by setting the other $7$ bits which will be permuted to the $i_1$-th pixel location, i.e., $i_1 =\lfloor u^{-1}(j)/8 \rfloor +1$, to $1$.

Referring to Corollary~1, $s(i_1)$ and $u(i_1)$, elements of the equivalent key streams of $S$ and $T$, can be determined by these equations.
Simultaneously, the mapping between $j$ and $k_m$ ($m \in [1, 8]$) can be also identified by checking the
bijection $y_n(i_1) \leftrightarrow k_m$ ($n \in [1,8]$).
Repeating this test for all the $L-1$ pixels, the relationship between $j\in [1, 8L]$ and $k \in [1, 8]$ can be totally revealed together
with the equivalent form of $S$ and $T$. What is more, we conclude that the data complexity involved is $(8L+1)$
in terms of number of chosen plain-images, which is linear to the size of the plain-image.

\subsection{Revealing the equivalent diffusion key stream $R$}
After recovering the permutation key stream $U$ and the equivalent substitution key streams $S$ and $T$, ICMPD becomes
a diffusion-only cipher that governed by Eq.~(\ref{eq:diffusion}).
Rewrite Eq.~(\ref{eq:chosenall0}) as
\begin{equation*}
r(i) =  t(i)  \oplus r(i) \oplus c(i) \oplus c(i-1),
\end{equation*}
then one can calculate the key stream $R$ using the chosen plain-image with fixed bit value $0$ and its corresponding cipher-image.
Finally, it is concluded that ICMPD can be broken at the cost of $8L+1$ chosen plain-images and their corresponding cipher-images.

To verify our analysis, we set the secret key to
\begin{small}
\small{$(0.346, 0.478, 1.644, 2.986, 4.434, 0.6435, 5.673, 0.523, 3.14, 0.34)$}
\end{small}
 and carry out some experiments to images of size $128 \times 128$.
Based on the assumption that the encryption machine can be temporarily accessed, we encrypt an image with all the pixels identical to zero.
Then, we consecutively modify the value of $128\times 128 \times 8$ bits of this zero image and obtain the corresponding $128\times 128 \times 8$ cipher-images.
The (equivalent) key streams $U$, $S$, $T$ and $R$ are deduced using the method described above. Then they are used to break the cipher-images shown in Fig.~\ref{fig1:cipher1}
and Fig.~\ref{fig1:cipher2}. The recovered result is depicted in Fig.~\ref{fig3:recover1} and Fig.~\ref{fig3:recover2}, which coincides with the original plain-images shown in Fig.~\ref{fig1:plain1}
and Fig.~\ref{fig1:plain2}.

\section{Discussion and conclusion}
In this paper, we have evaluated a new image cryptosystem based on modified permutation-diffusion architecture \cite{zhu2014image} in a chosen plaintext attack scenario.
As we claimed, the reason for the successful implementation of our CPA scheme is twofold:
a)  the iteration round of the permutation-diffusion round is merely one;
b) the key schedule is independent from the encryption process.
In concern to these problems, a simple remedy is to increase the iteration round \cite{fridrich1998symmetric, chen2004symmetric} based on
a comprehensively quantitative study on the tradeoff between complexity and security.
An alternative solution is to embed some feedback mechanism in the key schedule \cite{zhang2014chaotic}, such that the whole cryptosystem
will operate in a supposedly one-time-pad manner. Thus the difficulty of the CPA analysis increases dramatically.

The goal of this paper is not to simply present our CPA method on a given image cryptosystem, but build a new framework to quantitatively
study the security level of classical modulo then XORing operation and then apply this result to a new diffusion kernel.
In this regard, the work shown in this paper would benefit the measure of security of image cryptosystem based on permutation-diffusion architecture,
and thus the designing of practical schemes.

\bibliographystyle{elsarticle-num}
\bibliography{Ref}

\newpage
\begin{figure}[!htb]
\centering
\begin{minipage}[t]{\imagewidth}
\centering
\includegraphics[width=\imagewidth]{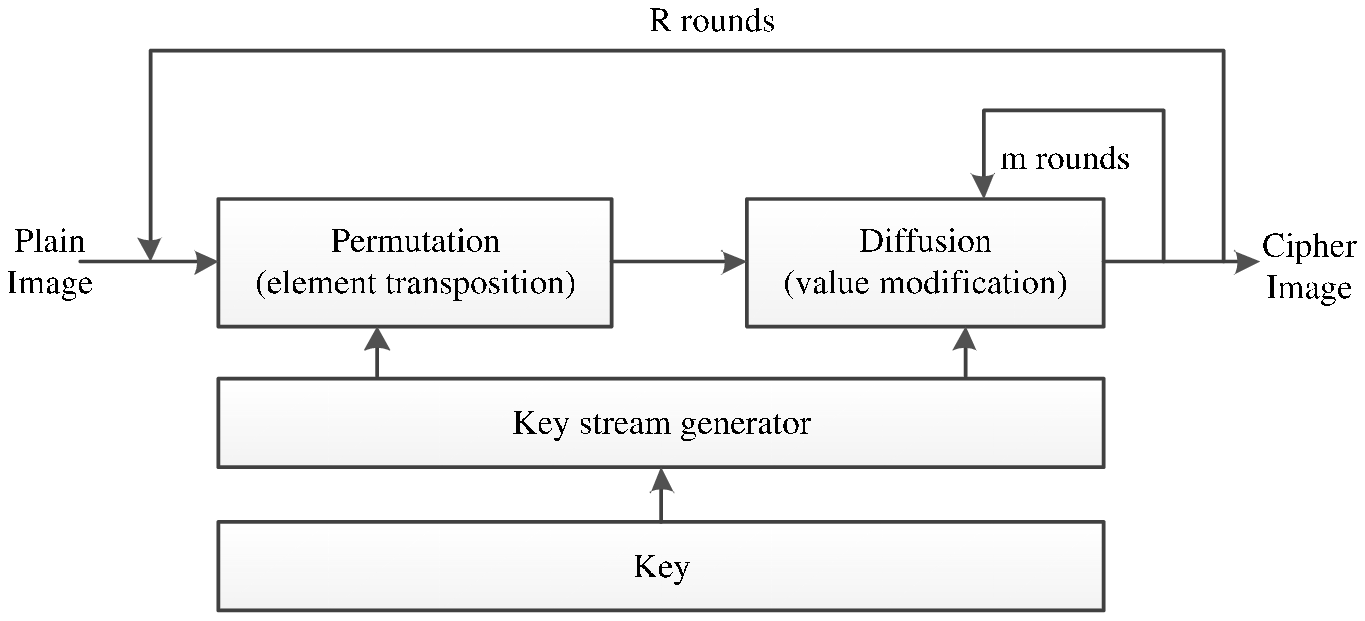}
\end{minipage}
\caption{Block diagram of the permutation-diffusion structure proposed by Fridrich.}
\label{fig:struture}
\end{figure}

\newpage
\begin{figure}[!htb]
\centering
\begin{minipage}[t]{\imagewidth}
\centering
\includegraphics[width=\imagewidth]{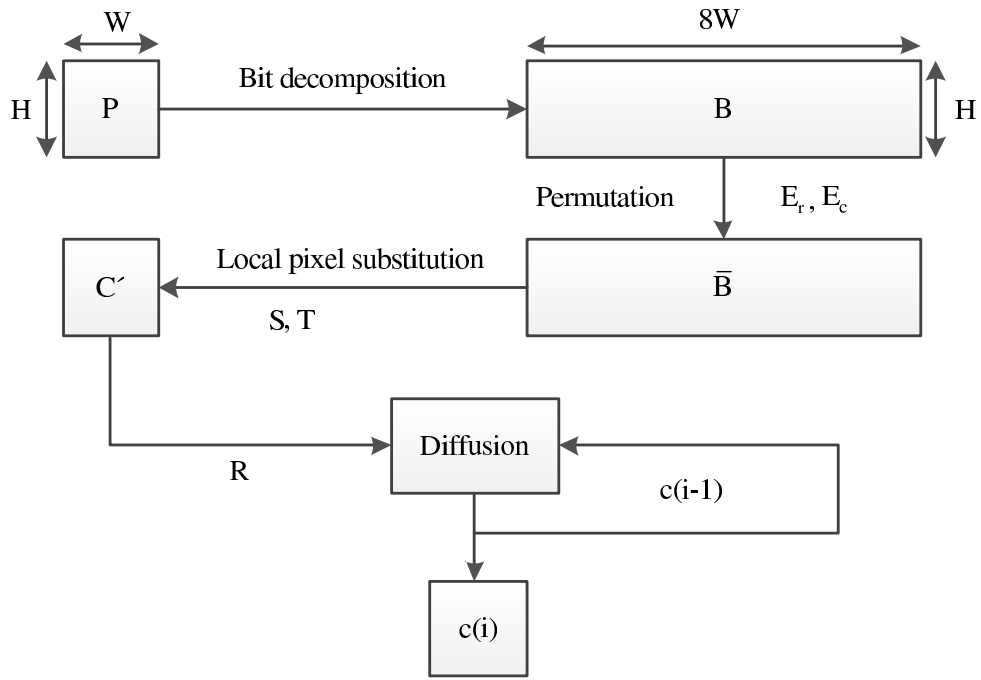}
\end{minipage}
\caption{Schematic diagram of the modified permutation-diffusion structure of ICMPD.}
\label{fig:newstruture}
\end{figure}

\newpage
\begin{figure*}[!htb]
\centering
\subfigure[]{
    \label{fig1:plain1}
    \begin{minipage}[t]{\imagewidtha}
    \centering
    \includegraphics[width=\imagewidtha]{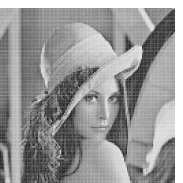}
    \end{minipage}}
\subfigure[]{
    \label{fig1:plain2}
    \begin{minipage}[t]{\imagewidtha}
    \centering
    \includegraphics[width=\imagewidtha]{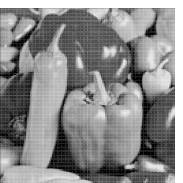}
    \end{minipage}}
\subfigure[]{
    \label{fig1:cipher1}
    \begin{minipage}[t]{\imagewidtha}
    \centering
    \includegraphics[width=\imagewidtha]{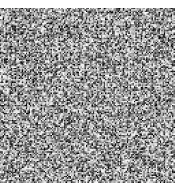}
    \end{minipage}}
\subfigure[]{
    \label{fig1:cipher2}
    \begin{minipage}[t]{\imagewidtha}
    \centering
    \includegraphics[width=\imagewidtha]{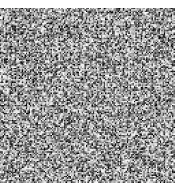}
    \end{minipage}}
\caption{Two plain-images and their corresponding cipher-iamges:
(a) plain-image ``Lena" of size $128\times 128$;
(b) plain-image ``Peppers" of size $128\times 128$;
(c) cipher-image corresponding to ``Lena";
(d) cipher-image corresponding to ``Peppers".}
\end{figure*}

\newpage
\begin{figure*}[!htb]
\centering
\subfigure[]{
    \label{fig:all0_c0}
    \begin{minipage}[t]{\imagewidtha}
    \centering
    \includegraphics[width=\imagewidtha]{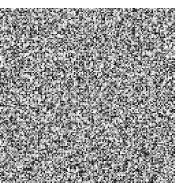}
    \end{minipage}}
\subfigure[]{
    \label{fig:all0_c1}
    \begin{minipage}[t]{\imagewidtha}
    \centering
    \includegraphics[width=\imagewidtha]{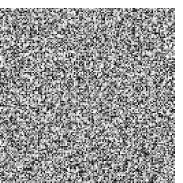}
    \end{minipage}}
\subfigure[]{
    \label{fig:all0_xor}
    \begin{minipage}[t]{\imagewidtha}
    \centering
    \includegraphics[width=\imagewidtha]{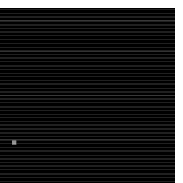}
    \end{minipage}}
\caption{Two cipher-image sequences and the difference between them:
(a) cipher-image sequence $V$ corresponding to $B$ ;
(b) cipher-image sequence $\tilde{V}_1$ corresponding to $\tilde{B}_1$;
(c) XOR between $V$ and $\tilde{V}_1$ (for perceptual purpose, we artificially set the value of the pixels around the non-zero one to $128$).}
\end{figure*}

\newpage
\begin{figure}[!htb]
\centering
\subfigure[]{
    \label{fig3:recover1}
    \begin{minipage}[t]{\imagewidtha}
    \centering
    \includegraphics[width=\imagewidtha]{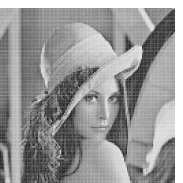}
    \end{minipage}}
\subfigure[]{
    \label{fig3:recover2}
    \begin{minipage}[t]{\imagewidtha}
    \centering
    \includegraphics[width=\imagewidtha]{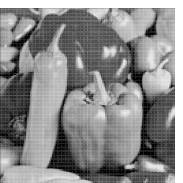}
    \end{minipage}}
\caption{Application example of our chosen plaintext attack:
(a) Recovered result from image shown in Fig.~\ref{fig1:cipher1} using the obtained equivalent key streams;
(b) Recovered result from image shown in Fig.~\ref{fig1:cipher2} using the obtained equivalent key streams.}
\end{figure}

\end{document}